\documentclass[a4paper,USenglish,cleveref,numberwithinsect]{lipics-v2021}

\usepackage{microtype} 
\usepackage{cite}
\usepackage{todonotes}

\graphicspath{{./figures/}}

\title{Robust Algorithms for Finding Triangles and 
Computing the Girth in Unit Disk and Transmission Graphs}
\titlerunning{Robust Algorithms for Unit Disk and Transmission Graphs}

\author{Katharina Klost}{Institut f\"ur Informatik, Freie Universit\"at Berlin, Germany}{kathklost@inf.fu-berlin.de}{}{}

\author{Wolfgang Mulzer}{Institut f\"ur Informatik, Freie Universit\"at Berlin, Germany}{mulzer@inf.fu-berlin.de}{https://orcid.org/0000-0002-1948-5840}{}
 
\authorrunning{K.~Klost and W.~Mulzer}



\hideLIPIcs
\nolinenumbers
\begin{document}

\maketitle

\begin{abstract}
We describe optimal \emph{robust} algorithms for finding 
a triangle and the unweighted girth in a unit disk graph, 
as well as finding a triangle in a transmission graph.
In the robust setting, the input is not given as 
a set of sites in the plane, but rather as an abstract graph.
The input may or may not be realizable as a unit disk graph 
or a transmission graph.
If the graph is realizable, the 
algorithm is guaranteed to give the correct answer. If not, 
the algorithm will either give a correct answer or correctly state 
that the input is not of the required type.
\end{abstract}

\section{Introduction}

Suppose we are given a set $S \subseteq \mathbb{R}^2$ 
of $n$ \emph{sites} in the plane, where each site $s \in S$ has an 
\emph{associated radius} $r_s > 0$. The \emph{disk graph} $D(S)$ 
on $S$ is defined as $D(S) = (S, E)$, where 
$E = \{ st \mid \| st \| \leq r_s + r_t\}$, with $\| st\|$ 
being the Euclidean distance between the sites $s$ and $t$. 
If all associated radii are $1$, $D(S)$ is called a 
\emph{unit disk graph}.
The \emph{transmission graph} on $S$
is the \emph{directed} graph with vertex set $S$ and a directed
edge $st$ from site $s$ to site $t$ if and only
if $\| st \| \leq r_s$, i.e., if and only if the $t$ 
lies inside the disk of radius $r_s$ centered at $s$. 
If all radii are equal, the edges $st$  and $ts$ are 
always either both present or both absent, for any two sites
$s, t \in S$, and the resulting 
transmission graph is equivalent to a unit disk graph.
Thus, for transmission graphs, the interesting case is 
that the associated radii are not all the same.

There is a large body of literature on (unit) disk graphs, see, 
e.g.,~\cite{cabello_shortest_2015,clark_unit_1990,graf_coloring_1998,
kaplan_dynamic_2021,kaplan_triangles_2019}.
Transmission graphs are not as widely studied, but recently they
have received some 
attention~\cite{kaplan_triangles_2019,kaplan_spanners_2015}.
Even though disk graphs and transmission graphs may have up to
$\Omega(n^2)$ edges, they can be described
succinctly with $O(n)$ numbers, namely the coordinates of the
sites and the associated radii. Thus, the underlying geometry
often makes it possible to 
find efficient algorithms whose running time depends only on $n$.

In the setting where a unit disk graph or a transmission graph 
is given as an abstract graph, 
not much is known. 
One possible explanation is that the problem of deciding 
whether an abstract graph is a (unit) disk graph or a transmission graph 
is $\exists\mathbb{R}$-hard~\cite{kang_sphere_2012,klost_complexity_2017}.
In fact, Kang and M\"uller~\cite{kang_sphere_2012} show that 
there are unit disk graphs whose coordinates need 
an exponential number of bits in their representation, so even 
if it is known that the input is a unit disk graph, it 
is not clear that a realization of the graph can be efficiently computed.
As transmission graphs with unit radii are equivalent to unit disk graphs, 
the result carries over to transmission graphs.

Raghavan and Spinrad~\cite{raghavan_robust_2003} introduced 
a notion of \emph{robust} algorithms in \emph{restricted domains}.
A restricted domain is a subset of the possible inputs. In 
our case, it will be the domain of unit disk graphs or transmission graphs 
as a subdomain of all abstract graphs.
Contrary to the \emph{promise} setting, in which the algorithm 
only gives guarantees for inputs from the restricted domain, 
the output in the robust setting must always be useful.
If the input comes from the restricted domain, 
the algorithm must always return a correct result.
If the input is not from the restricted domain, the algorithm 
may either return a correct result, or correctly state that 
the input does not meet the requirement.
Raghavan and Spinrad~\cite{raghavan_robust_2003} give a 
robust polynomial-time algorithm for the \textsc{Clique}-problem 
in unit disk graphs.

The problem of finding a triangle or of computing the \emph{girth}
(the shortest unweighted cycle) in a graph is a basic 
algorithmic question in graph theory.
The best know algorithm for general graphs uses matrix multiplication 
and runs in either $O(n^\omega)$ or $O(n^{2\omega/(\omega + 1)})$ time, 
where $\omega \leq 2.371552$ is the matrix multiplication 
constant~\cite{alon_finding_1997,gall_powers_2014-arxiv,itai_finding_1978,
vassilevskawilliams_new_2024}.
The best combinatorial algorithm needs 
$O\left(n^3/2^{\Omega(\sqrt[7]{\log n})}\right)$ 
time~\cite{yu_improved_2015-arxiv,abboud_new_2023}.

For special graph classes, better results are known.
In the case of \emph{planar} graphs, Itai and 
Rodeh~\cite{itai_finding_1978}, and, independently, 
Papadimitriou and Yannakakis~\cite{papadimitriou_clique_1981} 
show that a triangle can be found in $O(n)$ time, if it exists.
Chang and Lu~\cite{chang_computing_2013} give an $O(n)$ time algorithm 
for computing the unweighted girth in an undirected planar graph.
Kaplan et al.\@\cite{kaplan_triangles_2019} show that in the 
geometric setting, finding a triangle and computing the unweighted 
girth can be done in $O(n\log n)$ time for general disk graphs and in 
$O(n\log n)$ expected time for transmission graphs. 
They also give algorithms with the same expected running time, for 
finding the smallest weighted triangle in disk graphs and transmission 
graphs, as well as for computing the weighted girth of a disk graph.
In the geometric setting, there are $\Omega(n\log n)$ 
lower bounds for finding (short) triangles and computing the 
(weighted) girth in the algebraic decision tree 
model~\cite{klost_geometric_2021,Polishchuk17}.

In this paper, we show that there are $O(n)$ 
time algorithms for finding a triangle and computing 
the girth in unit disk graph, in the robust setting. 
Furthermore, we extend the ideas to an algorithm for finding a 
triangle in a transmission graph in $O(n + m)$ time.
The running times for the algorithms in the unit disk graph setting 
can be sublinear in the input size, as the input in the robust setting 
consists of a representation of all vertices and edges.
In particular, the result is better than the $\Omega(n\log n)$ 
lower bound for the geometric setting, because this lower bound stems
from the difficulty of finding the edges of the graph.
The running time for transmission graphs is linear in the input size, and
it is significantly faster than the currently fastest algorithm 
for general graphs~\cite{kang_sphere_2012}.

\section{Preliminaries}

We assume that the input is an abstract unweighted graph 
$G = (V, E)$, given as an adjacency list.
For an undirected graph, given a vertex $v$, we denote by 
$N(v)$ the set of vertices that are adjacent to $v$, and 
by $\deg(v) = |N(v)|$ the degree of $v$.
In the adjacency list representation, a set of $k$ neighbors of $v$ 
can be reported in $O(k)$ time, and testing if two vertices $u$ and $v$ 
are adjacent takes $O(\min(\deg(v), \deg(u)))$ time.
For a directed graph, let $N_{\text{in}}(v)$ and $N_{\text{out}}(v)$ 
be the vertices connected by an incoming or outgoing edge to a vertex $v$, 
respectively. Let $N_{\text{bi}} = N_{\text{in}} \cap N_{\text{out}}$ 
be the vertices connected by both an incoming and an outgoing edge.

\begin{definition}[Raghavan and Spinrad~\cite{raghavan_robust_2003}]
A \emph{robust} algorithm for a problem $P$ on a domain $C$ solves 
$P$ correctly, if the input is from $C$.
If the input is not in $C$, the algorithm may either produce a correct answer 
for $P$ or report that the input is not in $C$.
\end{definition}

\section{Unit disk graphs}

Our algorithms make use two key properties. Both properties 
are well known. For completeness, we include proofs for the variants we use.

\begin{lemma}\label{lem:deg_bound}
Let $G = (V, E)$ be a graph that is realizable as an unit disk graph, 
and let $v \in V$ be a vertex with $\deg(v) > 5$.
Then, the subgraph induced by $v$ and any six adjacent 
vertices contains a triangle.
\end{lemma}

\begin{proof}
Consider a realization of the unit disk graph in the plane. 
We identify the vertices with the 
corresponding sites.  Let $u_0, \dots, u_5 \in N(v)$ be any
six neighbors of $v$, labelled in clockwise order
around $v$. 

Let 
$\alpha_i = \sphericalangle u_i v u_{i+1}$, $i = 0, \dots 5$, be the
angles between the consecutive neighbors with respect to $v$, where 
the indices are taken modulo $6$.
Note that $\sum_{i = 0}^5 \alpha_i = 2\pi$,
and suppose that 
$\alpha_0$ is a minimum angle in this sequence.

First, suppose that  $\alpha_0 < \pi/3$. Then, there is a cone with opening angle 
$\pi /3$ and apex $v$ that contains the sites $u_0$ and $u_1$. 
Since $u_0$ and $u_1$ are adjacent to $v$,  
we have $\| u_0v\| \leq 2$ and  $\| u_1v \| \leq 2$, so $u_0$ and $u_1$ 
both lie inside a circular sector with angle $\pi / 3$, apex $v$, and 
radius $2$.
This circular sector has diameter $2$. Thus, all points in it, in particular 
$u_0$ and $u_1$, have mutual distance at most $2$. If follows that $u_0$ and $u_1$ are 
connected by an edge, closing a triangle, see \autoref{fig:udg_deg}, left.

Second, suppose that $\alpha_0 \geq \pi/3$. Then, since $\alpha_0$ is minimum
and since the $\alpha_i$ sum to $2\pi$, 
it follows that  $\alpha_i = \pi/3$, for $i = 0, \dots, 5$, so the $u_i$ lie on
six concentric, uniformly spaced rays that emanate from $v$.
In this case, the maximum possible mutual distance between the $u_i$'s is achieved 
when the sites constitute the corners of a regular hexagon with center $v$ 
and $\|u_iv\| = 2$, for $i = 0, \dots, 5$. 
This hexagon decomposes into six equilateral triangles of side length $2$, 
and thus the unit disk graph contains all consecutive edges $\{u_i, u_{i+1}\}$, 
closing a triangle, see \autoref{fig:udg_deg}, right.
\begin{figure}
\centering
\includegraphics{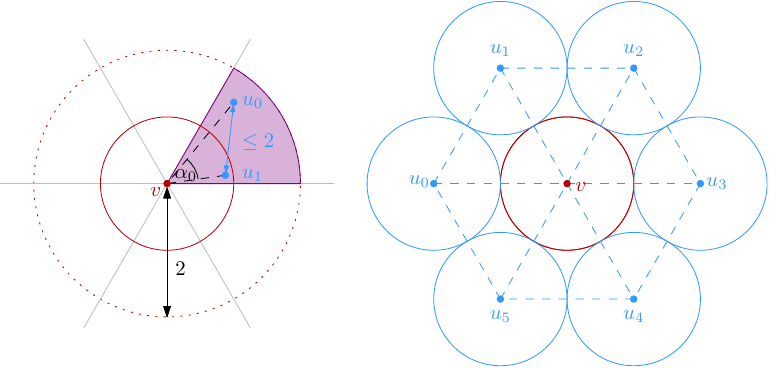}
\caption{Left: The disk defined by $v$ is marked in solid red. 
There are two sites in the shaded sector. The diameter of the sector is $2$,
so $u_0$ and $u_1$ are connected by an edge.\\
Right: The mutual distance between the $u_i$ is maximized on the vertices of a regular polygon. 
	The side length of each equilateral triangle is $2$, and thus the desired edges exist.}
\label{fig:udg_deg}
\end{figure}
\end{proof}

\begin{lemma}\label{lem:non-triangle-planar}
If a (unit) disk graph does not contain a triangle, it is planar.
\end{lemma}

\begin{proof}
Given the geometric representation of a unit disk graph, 
the natural embedding of this unit disk graph into the plane is 
by connecting the sites that represent the vertices by line segment.
Evans et al.\@\cite{evans_recognizing_2016}, as well as 
Kaplan et al.\@\cite{kaplan_triangles_2019} show that if this 
embedding is not plane, there has to be a triangle in the unit disk graph.
Conversely, this implies that if there is no triangle in the graph, 
the natural embedding is crossing free, directly implying 
that the graph is planar.
\end{proof}

\subsection{Finding a Triangle}\label{subsec:udgtriangle}

\begin{theorem}\label{thm:triangle}
There is a robust algorithm to find a triangle in a disk graph in 
$O(n)$ time.
\end{theorem}
\begin{proof}
The algorithm works as follows.
If there is no vertex $v$ with $\deg(v) > 5$, 
then we check explicitly for every vertex whether two of its neighbors 
are adjacent. If so, we have found a triangle. If not, there is none.
As all degrees are constant, this takes $O(1)$ time per vertex, for a total 
of $O(n)$ time.

Now, assume there is a vertex $v$ with $\deg(v) >5$. 
Let $N'(v)$ be a set of any seven neighbors of $v$.
For every pair of neighbors $u$, $w$ from $N'(v)$, explicitly check if 
there is an edge between $u$ and $w$.
If an edge is found, report the triangle $u,v,w$.
Otherwise, report that the input is not a unit disk graph.
This step takes $O(n)$ time to identify $v$ and then at most $O(n)$ 
time to check the adjacencies for each of the $O(1)$ vertices in $N'(v)$,
summing up to $O(n)$ total time.
Note that in the case that not all degrees are at most five, only 
one vertex is considered in detail.

To see that the algorithm is correct, we consider all possible cases. 
If the maximum degree of the graph is at most $5$, all 
vertices and their neighbors are explicitly checked. So if there is a triangle 
in the graph, the algorithm will find it and correctly report it. 
Furthermore, no triangles can be missed.
Otherwise, there is a vertex $v$ with degree larger than $5$. 
If the input is a disk graph, \autoref{lem:deg_bound} guarantees that 
there is a triangle in $N'(v)$.
The algorithm explicitly searches for such edge between vertices of \(N(v')\). 
If such an edge is found, the triangle is correctly reported.
In the other case, \autoref{lem:deg_bound} implies that the input is 
not a unit disk graph, as reported by the algorithm.
\end{proof}

\subsection{Computing the Girth}

\begin{theorem}
There is a robust algorithm to compute the girth of a 
graph in the domain of unit disk graphs that runs in $O(n)$ time.
\end{theorem}

\begin{proof}
First, run the algorithm from \autoref{thm:triangle} on the input. 
If the algorithm determines that the input graph is not a unit disk graph, 
report this and finish.
If the algorithm found a triangle, the girth of the graph is three 
and can be reported.

If the algorithm from \autoref{thm:triangle} did not find a triangle and 
did not report that the graph is not a unit disk graph, we 
use a linear time planarity testing algorithm on the graph,
e.g., the algorithm described by Hopcroft and 
Tarjan~\cite{hopcroft_efficient_1974}.
If the graph is not planar, report that it is not a unit disk graph.
In the other case, the algorithm for computing the girth in a planar graph by 
Chang and Lu~\cite{chang_computing_2013} can be used to compute the girth 
of the graph in $O(n)$ time.

By combining the running times for each step, the overall running time 
follows.
Correctness follows from \autoref{lem:non-triangle-planar} and 
the correctness of the algorithm by Chang and Lu.
\end{proof}

\section{Finding a directed triangle in a transmission graph}

The following key lemma needed for the robust algorithm for 
transmission graphs was previously shown by Klost~\cite{klost_geometric_2021}.
We include a proof for completeness.

\begin{lemma}\label{lem:tg_bounded}
If $G$ is a transmission graph and $v$ is a vertex with 
$|N_\text{bi}(v)| > 6$, then $G$ contains a triangle.
\end{lemma}

\begin{proof}
The proof is similar to that of \autoref{lem:deg_bound}. 
Consider a representation of $G$ in the plane and the disk $D_v$ 
associated to the site $v$. By definition, all sites in 
$N_\text{bi}(v)$ lie in $D_v$.
Subdivide $D_v$ into six congruent circular sectors, as in
\autoref{fig:tg_deg}. By the pigeonhole principle, 
one sector $C$ contains at least two vertices $u$ and $w$ from 
$N_\text{bi}(v)$.

\begin{figure}
\centering
\includegraphics{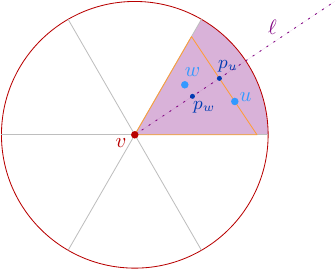}
\caption{The disk defined by $v$ is marked in solid red. 
There are two sites in the shaded sector. The distance between two points 
in the triangle is at most $\| uv \|$.}
\label{fig:tg_deg}
\end{figure}

Let $\ell$ be the perpendicular bisector of $C$ and let $p_u$ and $p_w$ 
be the perpendicular projections of $u$ and $w$ onto $\ell$.
W.l.o.g., suppose that $p_w$ is closer to $v$ than $p_u$.
Consider the equilateral triangle $\triangle$ defined by the line 
through $u$ perpendicular to $\ell$ and the rays defining $C$.
Then, $w$ is contained in $\triangle$, and thus 
$\| uw \| \leq \| uv\| \leq r_u$. Hence, the edge $uw$ 
exists in the transmission graph, closing the directed triangle.
\end{proof}

\begin{lemma}\label{lem:find_neighbors}
  The sets $N_\text{bi}(v)$, for all $v \in V$ can be found in $O(n + m)$ time.
\end{lemma}
\begin{proof}
  We assume that the adjacency list representation has the 
  following standard form: on the top level, there is an array that 
  can be indexed by the vertices in $V$.
  Each entry points to the head and the tail of a linked list. 
  This linked list contains all vertices $u$ such that $(v, u)\in E$, in no particular order.
  Note that the order of the array directly induces a total order $\preceq$ on $V$.

  We compute adjacency list representations $L_\preceq^\top$ and $L_\preceq$ 
  of the transposed graph $G^\top$ and the original graph $G$, with the additional 
  property that linked list are sorted according to $\preceq$.
  For this, we initialize a new array and traverse the original adjacency list 
  representation of $G$, by considering the vertices according to $\preceq$.
  For every edge $(v, u)$ that is encountered, we append $v$ to the linked list of $u$, 
  in $O(1)$ time.
  As the source vertices $v$ are traversed according to $\preceq$, this gives the 
  desired representation $L_\preceq^\top$ of $G^\top$, in time $O(n + m)$.
  After that, we can obtain the representation $L_\preceq$ of $G$ by the same procedure,  
  using $L_\preceq^\top$ as the initial adjacency list.
  Finally, to identify the sets $N_\text{bi}(v)$, for each $v \in V$, it suffices
  to merge the associated lists $L_\preceq$ and $L_\preceq^\top$, in total time
  $O(n + m)$. 
\end{proof}

\begin{observation}\label{lem:cycle_biedge}
In every directed cycle \(C\) in a transmission graph, there is least one vertex \(v\) with \(N_\text{bi}(v) \neq \emptyset\).
\end{observation}
\begin{proof}
  Let \(v\) be the site with the smallest radius on \(C\) and let \(u\) be the successor of \(v\) on \(C\).
  Then the edge \((v,u)\) exists by definition. Furthermore, the existence of this edge directly implies that \(\Vert uv \Vert \leq r_v \leq r_u\) and thus \((u,v)\) is also an edge of the transmission graph.
\end{proof}

\begin{theorem}
There is a robust algorithm that finds a directed 
triangle in a transmission graph in $O(n + m)$ time.
\end{theorem}

\begin{proof}
Preprocess the input as described in the proof of \autoref{lem:find_neighbors} in $O(n + m)$ time, 
such that the set $N_\text{bi}(v)$ is known for every vertex $v$. This also gives the representations \(L_\preceq\) and \(L_\preceq^\top\) which can be used to compute a subgraph \(G_\text{uni}\) of \(G\) by removing all bidirected edges.

We consider two cases.
In the first case, all vertices have $|N_\text{bi}(v)|\leq 6$.
Test in \(O(n+m)\) time if \(G_\text{uni}\) is acyclic.
If yes, report that \(G\) is not a transmission graph.
If not, consider for each vertex \(v\) the  edges \((u,v)\) with \(u\in N_\text{bi}(v)\) explicitly and check if there is a vertex \(w \) in \(N_\text{in}(v)\cap N_\text{out}(u) \neq \emptyset\). If yes, report the triangle $vuw$.
If no triangle was found after considering all vertices, report that \(G\) does not have a triangle.
This step take $O(n+m)$ overall time as the adjacency list of every vertex is traversed a constant number of times.
By \autoref{lem:cycle_biedge}, the first check makes sure that no triangles without bidirected edges are missed, as in this case \(G_\text{uni}\) is not acyclic, and it is correctly reported that the graph is not a transmission graph.
In the second step, all possible triangles with at least one bidirected edge are explicitly checked.

In the second case, let $v$ be a vertex with $|N_\text{bi}(v)| > 6$. 
Find a set of $7$ vertices from this set and explicitly check 
if any pair of them closes a triangle.
If there is no triangle, then report that the input 
is not a transmission graph. 
Otherwise, report the triangle.
This again takes $O(n + m)$ time. The correctness follows directly from \autoref{lem:tg_bounded}.
\end{proof}

\section{Conclusion}

We showed that there are robust sublinear 
algorithms for finding a triangle and computing the girth 
in unit disk graph as well as a linear time algorithm for finding 
a triangle in a transmission graph.
Extending the arguments to general disk graphs seems to be hard, 
as the properties given in \autoref{lem:deg_bound} and
\autoref{lem:tg_bounded} do not easily carry over to general disk graphs.
It would be interesting to see if there are properties of transmission 
graphs that allow a sublinear running time similar to the unit disk graph case.

\bibliography{girth}

\end{document}